\documentclass[10pt,conference]{ieeeconf}  

\IEEEoverridecommandlockouts             
\usepackage{cite}
\usepackage{balance}
\usepackage{amsmath,amssymb,amsfonts,bm}
\usepackage{mathrsfs}

\usepackage{graphicx}
\usepackage{textcomp}
\usepackage[capitalise]{cleveref}
\usepackage{float}

\usepackage{fp,tikz,pgfplots}
\usetikzlibrary{arrows,shapes,backgrounds,patterns,fadings,matrix,arrows,calc, intersections,decorations.markings, positioning,arrows.meta}
\usepgfplotslibrary{fillbetween}
\usepgfplotslibrary{statistics}
\pgfplotsset{width=5\columnwidth /5, compat = 1.13,
	height = 60\columnwidth /100, 
	legend cell align = left, ticklabel style = {font=\scriptsize},
	every axis label/.append style={font=\small},
	legend style = {font={\scriptsize}},title style={yshift=-7pt, font = \small} }
\usepackage{pgfkeys}
\usetikzlibrary{tikzmark}
\usetikzlibrary{matrix,calc}
\usetikzlibrary{graphs,arrows.meta}
\pgfplotsset{compat=newest}
\usetikzlibrary{plotmarks}
\usepgfplotslibrary{patchplots}
\usepackage{grffile}
\usetikzlibrary {datavisualization}
\usetikzlibrary {datavisualization.formats.functions}
\usetikzlibrary {patterns}
\usepackage[T1]{fontenc}
\usepackage{helvet}
\newtheorem{assumption}{Assumption}
\newtheorem{definition}{Definition}
\newtheorem{lemma}{Lemma}
\newtheorem{theorem}{Theorem}

\usepackage{diagbox}

\usepackage{glossaries}
\setacronymstyle{long-short}
\newacronym{gp}{GP}{Gaussian process}
\newacronym{mas}{MAS}{multi-agent system}
\newacronym{elmas}{ELMAS}{Euler-Lagrange multi-agent system}

\usepackage{algorithm}
\usepackage{algpseudocode}

\usepackage{caption}
\usepackage{listings} 

\usepackage{subcaption}

\newcommand{\CASE}[1]{\STATE \textbf{case} #1\textbf{:} \begin{ALC@g}}
\newcommand{\ENDCASE}{\end{ALC@g}}

\newcommand{\DEFAULT}{\STATE \textbf{default:} \begin{ALC@g}}
\newcommand{\ENDDEFAULT}{\end{ALC@g}}
\newcommand{\DEFAULTLINE}[1]{\STATE \textbf{default:} }

\definecolor{mycolor1}{RGB}{166,206,227}
\definecolor{mycolor2}{RGB}{31,120,180}
\definecolor{mycolor3}{RGB} {178,223,138}
\definecolor{mycolor4}{RGB}{51,160,44}

\graphicspath{{tikz/}}

\title{\LARGE \bf
Cooperative Learning with Gaussian Processes for Euler-Lagrange Systems Tracking Control under Switching Topologies 
}

\author{Zewen Yang$^{1,2}$, Songbo Dong$^{2}$, Armin Lederer$^{2}$, Xiaobing Dai$^{2}$, Siyu Chen$^{3}$, Stefan Sosnowski$^{2}$, \\Georges Hattab$^{1,4}$ and Sandra Hirche$^{2}$ 
 \thanks{Preprint.}
\thanks{$^{1}$ Robert Koch Institute, Berlin, Germany. $^{2}$ Technical University of Munich, Munich, Germany. $^{3}$ The University of Texas at Austin, Austin, TX, USA. $^{4}$ Freie Universität Berlin, Berlin, Germany.
{Correspondence to: Zewen Yang $<$yangz@rki.de$>$}}%
}

\begin{document}

\maketitle
\thispagestyle{empty}
\pagestyle{empty}

\begin{abstract}
This work presents an innovative learning-based approach to tackle the tracking control problem of Euler-Lagrange multi-agent systems with partially unknown dynamics operating under switching communication topologies. The approach leverages a correlation-aware cooperative algorithm framework built upon Gaussian process regression, which adeptly captures inter-agent correlations for uncertainty predictions. 
A standout feature is its exceptional efficiency in deriving the aggregation weights achieved by circumventing the computationally intensive posterior variance calculations.
Through Lyapunov stability analysis, the distributed control law ensures bounded tracking errors with high probability. Simulation experiments validate the protocol's efficacy in effectively managing complex scenarios, establishing it as a promising solution for robust tracking control in multi-agent systems characterized by uncertain dynamics and dynamic communication structures.
\end{abstract}


\section{Introduction}
Multi-agent systems (MASs) have attracted significant attention within the field of control due to their ability to collaboratively achieve overarching objectives \cite{ohSurveyMultiagentFormation2015}. Although many studies concentrate on linear agent dynamics, the application of control methods developed for linear systems is inadequate for intricate physical systems. In this paper, we delve into Euler-Lagrange MASs, which serve as a modeling framework for diverse physical systems like robotic manipulators \cite{gao2021quasi} and underwater vehicles \cite{yanNovelPathPlanning2018}.

Accomplishing intricate tasks in MASs often involves extensive investigations of tracking control protocols, as evident in numerous studies \cite{yangDistributedGlobalOutputFeedback2017, heLeaderFollowingConsensusMultiple2021,YAN2019361,yanVirtualLeaderBased2020}. 
Many of these approaches assume a prior understanding of system dynamics and environmental disturbances, which poses a significant constraint when dealing with uncertain MASs operating in unfamiliar environments. To overcome this challenge, there has been a surge in interest toward learning-based control methods that leverage collected data to infer uncertainties inherent to the environment. Particularly in the context of safe control tasks amidst uncertainties, Gaussian process regression (GPR) \cite{rasmussenGaussianProcessesMachine2006} has emerged as a popular choice for modeling the effects of unknown environmental factors on system dynamics. GPR's appeal lies in its robust expressive capabilities, a theoretical foundation that accommodates statistical prediction error bounds, and an intrinsic trade-off between bias and variance \cite{umlauftFeedbackLinearizationBased2020}. Nonetheless, the efficacy of predictions of GPR is notably sensitive to the quantity of available training data \cite{dai2023can}. Employing a large dataset can significantly escalate computational demands, thereby impeding the real-time applicability of GPR in control tasks.

Various techniques have been developed to mitigate the training and prediction complexity of GPR, which increases cubically with the number of training samples. These methods include inducing point techniques \cite{pmlr-v2-snelson07a}, finite feature approximations \cite{NEURIPS2018_4e5046fc}, and aggregation strategies \cite{deisenrothDistributedGaussianProcesses2015}, while these methods primarily focus on single-agent systems. 
The concepts from aggregation techniques have been extended to enable cooperative learning using GPR within MASs. 
In \cite{beckers2021online,dai2024DecentralizedEventTriggered}, the event-triggered learning-based incorporating GPR allows individual learning in each agent, such that it provides probabilistic guarantees for safe consensus control. 
Despite its advantages, this method overlooks the potential benefit of information exchange between locally learned models. 
The previous works \cite{yangDistributedLearningConsensus2021,ledererCooperativeControlUncertain2023,dai2024cooperative} propose cooperative learning approaches, where the agents aggregate the predictions from their neighboring agents. While this approach achieves accurate predictions, it necessitates the additional computations of Gaussian process posterior variances for determining aggregation weights or optimized parameters~\cite{Hoang_Hoang_Low_How_2019,heDistributedOnlineSparse2023}. Although \cite{yangAAMAS,yangAAMASextend} proposed elective learning for mitigating the computational burden of joint prediction, the proposed method requires prior knowledge, which may not be available for certain systems.
In this paper, we present a collaborative learning framework based on cooperative GPR offering computational efficiency, while still maintaining the established theoretical bound on tracking error for the MAS control.

The contribution of this paper is in the form of a fresh approach to cooperative learning for distributed control rooted in the GPR technique and is designed to address uncertainty in Euler-Lagrange multi-agent systems (ELMAS). The novel learning framework, named cooperative correlation-aware GP (Cora-GP), leverages established aggregation methods while bypassing the need for calculating GP posterior variances. We provide two computationally efficient strategies for realizing the Cora-GP approach and incorporate them into a distributed consensus tracking control law. The effectiveness of the resulting control laws is formally shown using convergence guarantees for the tracking error of the ELMAS and demonstrated numerically in simulations. Notably, this convergence is achieved in the proximity of the origin, even within semi-Markov switching communication topologies.

The remainder of this article is structured as follows: Preliminaries and the problem formulation are stated in \cref{sec_P}. In \cref{sec_DistributedLearning}, the novel correlation-aware GP approach is presented. The learning-based protocol for consensus tracking control of the ELMAS is proposed, and stability for the resulting closed-loop MASs is proven in \cref{sec_ConsensusTracking}. A numerical simulation demonstrates the effectiveness of the proposed approach in \cref{sec_Simulation}, followed by a conclusion.


\section{Preliminaries and Problem Formulation} \label{sec_P}

\subsection{Notation and Graph Theory}
We denote real positive numbers without/with zero as $\mathbb{R}_{+} / \mathbb{R}_{0,+}$, naturals without/with zero as $\mathbb{N} / \mathbb{N}_{0}$, respectively. If not stated otherwise, identity matrix, null vector and vector or matrix of elements 1 are denoted by $ \bm{I}, \mathbf{0}$ and $\mathbf{1}$ with appropriate size, respectively. The Euclidean norm of a vector or matrix is denoted by $\|\cdot\|$, the cardinality of a set $\mathcal{N}$ is represented as $\left | \mathcal{N} \right |$, and the Kronecker product is indicated by $\otimes$. Minimum/maximum singular values of a matrix are denoted by~$\underline{\sigma}(\cdot)$/$\bar{\sigma}(\cdot)$. Matrix $\bm{A} \succ {0}$, if $\bm{A}$ is a positive definite matrix. The operation $\mathrm{blkdiag}()$ returns a block diagonal matrix created by aligning the input matrices.

In this paper, we use a digraph $ \mathcal{G} =  ( \mathcal{V} ,\mathcal{E} ) $ to describe the communication among the EL agents, where $ \mathcal{V} = \{ 1,\dots, n \} $ denotes the set of nodes, and $ \mathcal{E}\subseteq \mathcal{V} \times \mathcal{V} $ denotes the set of edges. A directed edge $\left ( i, j \right ) $ indicates that the $i$-th agent receives the information from $j$-th agent. The weighted adjacency matrix of $\mathcal{G}$ is denoted by $\boldsymbol{A} = \left [a_{ij} \right ]\in  \mathbb{R}^{n\times n} $, where an adjacency entry $a_{ij}>0$ if $\left (j,i \right )  \in \mathcal{E} $ and $ a_{ij}=0 $ otherwise. Moreover, it is assumed that the diagonal entries of the matrix $\bm{A}$ are zero, which implies $a_{ii}=0,~ \forall i \in \mathcal{V}$. Furthermore, we define the self-loop included adjacency matrix $\check{\boldsymbol{A}}= \left [\check{a}_{ij} \right ]\in  \mathbb{R}^{n\times n} $ with entry $\check{a}_{ij} = 1$ if $a_{ij}>0 $, $\check{a}_{ii} = 1$ and $ a_{ij}=0 $ otherwise. The Laplacian matrix of a digraph is defined as ${\bm{L}} = {\bm{D}} -{\bm{A}}$, where ${\boldsymbol{D}}=\operatorname{diag}({d}_{11}, {d}_{22}, \ldots {d}_{n n})$ with $ {d}_{ii} = \sum_{j=1}^{n} {a}_{ij}$ is the degree matrix of graph ${\mathcal{G}}$. The set of neighbours of agent $i$ is represented by $\mathcal{N}_i= \{j\in \mathcal{V}:(j,i)\in {\mathcal{E}}\}$. Similarly, let $\bar{\mathcal{G}}=(\bar{\mathcal{V}},\bar{\mathcal{E}})$ be the digraph of the leader-follower agents with the node set $ \bar{\mathcal{V}} = \left \{ 0 \right \} \cup \mathcal{V}$ and the edges set $  \bar{\mathcal{E}}\subseteq \bar{\mathcal{V}} \times \bar{\mathcal{V}} $, where the virtual leader is denoted by node $0$. The Laplacian matrix of $\bar{\mathcal{G}}$ is denoted as
\begin{align}
    \bar{\bm{L}}= \begin{bmatrix}
        0 & \bm{0}_{1\times n}\\
        \bm{L}_0 & \tilde{\bm{L}}
    \end{bmatrix} \in \mathbb{R}^{(n+1)\times (n+1)}, \nonumber
\end{align}
where $\bm{L}_0 = [a_{i0}]_{i=1,\dots,n} \in\mathbb{R}^{n}$, $\tilde{\bm{L}} =  \tilde{\bm{D}} -{\bm{A}}$, and the diagonal matrix $\tilde{\bm{D}}=\operatorname{diag}( \tilde{d}_{11}, \ldots \tilde{d}_{n n})$ with $ \tilde{d}_{ii} = \sum_{j=0}^{n} a_{ij}$.


\subsection{Stochastic Communication Topology}
In this paper, we consider a class of time-varying stochastic topologies (graphs), which is described by a semi-Markov process. This class of time-varying topologies is based on a set of fixed topologies $\bar{\mathcal{G}}_r$, where the index $r$ belongs to the finite state space $\mathcal{P}=\{1,2,\ldots, N\}$ with well-defined $N \in \mathbb{N}$. Let $r(t)=r_k\in\mathcal{P}$, $t\in \mathbb{R}_{0,+}$, denote the index of the topology at the $k$-th time interval $t \in [t_k, t_{k+1})$, where $k\in\mathbb{N}_0$. Then, the sojourn (holding) time at state $r_k$ is denoted by $\tau_{k}=t_{k+1}-t_{k}$ \cite{guoScaledConsensusProblem2021}. Based on this notation, we define semi-Markov time-varying topologies as follows.

\begin{definition}
Consider a stochastic process $\{r(t)\}_{t\in \mathbb{R}_{0,+}}$ indexing a set of fixed topologies $\bar{\mathcal{G}}_{r(t)}$, where $r(t)\in\mathcal{P}$. Let the process have step-wise trajectories with jumps at times $t_k, k \in \mathbb{N}$ satisfying $0<t_{1}<t_{2}<\cdots<t_{n}<\cdots$, such that the sequence of topology indexes $r(t_k)$ satisfies the Markov property, i.e., the probability $\Pr\{r(t_{k+1})=r_{k+1} \mid r(t_0)=r_{0},~ r(t_1)=r_{1}, \cdots, r(t_k)=r_{k}\} = \Pr\{r(t_{k+1})=r_{k+1} \mid r(t_k)=r_{k}\}$ for all $r_{0},r_{1},\cdots,r_{k},r_{k+1} \in \mathcal{P}$.
Moreover, let the distributions of the holding time $\tau_{k}$ be described in terms of distribution functions $ F_{ij}(\tau_k)$ via $ \Pr\left\{t_{k+1}-t_{k} \leq \tau_k, r\left(t_{k+1}\right)=j \mid r\left(t_{k}\right)=i\right\}=\mathrm{P}_{ij}  F_{ij}(\tau_k)$. The probabilities $\mathrm{P}_{ij}=\Pr\{r(t_{k+1})=j \mid r(t_k)=i\}$ define a transition probability matrix $\mathbf{P}_r=[\mathrm{P}_{ij}]\in\mathbb{R}^{N\times N}$ with $\mathrm{P}_{ii} = 0$, where $i,j\in\mathcal{P}$. Then, $r(t)$ describes the topology indexes of a semi-Markov time-varying topology.
\end{definition}
The behavior of semi-Markov time-varying topologies can be intuitively described as follows. Once a topology $\bar{\mathcal{G}}_{r_k}$ is chosen, it remains constant for the sojourn time $\tau_k$. The sojourn time $\tau_k$ is a random variable itself with probability distributions $F_{ij}$, which depend on the current topology and the next topology. When the topology is switched, the new topology $\bar{\mathcal{G}}_{r_{k+1}}$ is sampled from the discrete probability distribution $\mathrm{P}_{ij}, j \neq i$. Since the edges in this stochastic communication topology change over time, it allows more realistic modeling of wireless communication networks, where the connection between two agents can break down. 

In order to ensure that a distributed algorithm coordinating the agents can work properly, the communication topology has to ensure sufficient connectivity among the agents over time \cite{dongTimeVaryingFormationTracking2017, guoScaledConsensusProblem2021}, which requires the following assumptions.

\begin{assumption}
\label{ass_switchinggraph}
At every time $t$, the communication topology $\bar{\mathcal{G}}_{r(t)}$ contains a spanning tree with the root node being the leader node 0.
\end{assumption}
Assumption \ref{ass_switchinggraph} ensures that the switching graph has at least a path from the leader to some agents. In switching systems, this assumption is common since it is essential for followers to track the leader\cite{huaAdaptiveLeaderFollowingConsensus2017}. In addition, an assumption for the transition probability matrix $\mathbf{P}_r$ needs to be imposed.

\begin{assumption} \label{ass_irr}
The transition probability matrix $\mathbf{P}_r$ is irreducible.
\end{assumption}
This Assumption \ref{ass_irr} ensures that all states intercommunicate, i.e., there exists a positive probability that allows transitioning between any pair of states within finite steps.

\subsection{Euler-Lagrange Multi-agent System}
In this paper, we consider an ELMAS consisting of $n$ homogenous follower agents, referred to as agents in the following, and one virtual leader. In particular, the dynamics of the $i$-th agent in the ELMAS is described as
\begin{equation}\label{dyn_agent}
        \boldsymbol{H}(\boldsymbol{q}_i) \ddot{\boldsymbol{q}}_i+\boldsymbol{C}(\boldsymbol{q}_i, \dot{\boldsymbol{q}}_i) \dot{\boldsymbol{q}}_i+\boldsymbol{g}(\boldsymbol{q}_i)+\boldsymbol{f}(\boldsymbol{p}_i)=\boldsymbol{u}_i,~i\in \mathcal{V},
\end{equation}
where $\boldsymbol{q}_i = [{q}_{1},{q}_{2}, \dots, {q}_m]^\top  \in \mathbb{X} \subset \mathbb{R}^{m}$ is the state of the $i$-th agent, $\boldsymbol{u}_i = [{u}_{1},{u}_2, \dots, {u}_m]^\top\in \mathbb{R}^{m}$ is the control input, and $\boldsymbol{p}_i:=[\boldsymbol{q}_i^\top,\dot{\boldsymbol{q}}_i^\top,\ddot{\boldsymbol{q}}_i^\top]^\top$. The functions $\boldsymbol{H}(\cdot): \mathbb{R}^{m}\rightarrow\mathbb{R}^{m\times m} $, {$\boldsymbol{C}(\cdot): \mathbb{R}^{m} \times \mathbb{R}^{m} \rightarrow\mathbb{R}^{m\times m} $ and $\boldsymbol{g}(\cdot): \mathbb{R}^{m}\rightarrow\mathbb{R}^{m}$ denote the inertia matrix, Coriolis matrix and the gravity vector. Since they can be easily identified using well-known techniques from robotics \cite{spongRobotModelingControl2020}, we assume them to be known in the sequel. The function $\boldsymbol{f}(\cdot)=[{f}_1(\cdot),\dots,{f}_m(\cdot)]^\top: \mathbb{R}^{3m}\rightarrow\mathbb{R}^{m} $ is assumed to be unknown, but identical in all agents. This setting can be found in a scenario where a homogeneous fleet of autonomous robots operates in an unknown environment interfering with the robot dynamics, e.g.,  hydrodynamic forces caused by ocean currents acting on underwater vehicles. The control task is to track a virtual leader, whose dynamics follows a prescribed reference trajectory $\boldsymbol{f}_r(\cdot):\mathbb{R}_{0,+}\rightarrow\mathbb{R}^m$, which yields \looseness=-1
\begin{equation} \label{dyn_laeder}
        \boldsymbol{q}_0 = \boldsymbol{f}_r(t),
\end{equation}
where $\boldsymbol{q}_0 $ is the state of the virtual leader. To ensure each agent can follow the leader, we pose the following assumption on the reference trajectory $\boldsymbol{f}_r(\cdot)$.
\begin{assumption}\label{ass_leader}
The reference trajectory $\bm{f}_r$ is at least twice continuously differentiable and $\|\dot{\bm{f}}_r(t)\|\leq \bar{f}_r$, $\bar{f}_r \in \mathbb{R}_{0,+}$.
\end{assumption}
This assumption is common for the control of Euler-Lagrange systems as it allows tracking a reference using control techniques such as feedback linearization or computed torque control\cite{ledererCooperativeControlUncertain2023}. Moreover, since the reference trajectory is a design choice, it is not restrictive in practice.

To infer a data-driven model of the unknown function $\bm{f}$, we assume the availability of measurements of $\bm{f}(\cdot)$ in each agent. These measurements satisfy the following conditions.

\begin{assumption}\label{ass_data}
    Each agent $i$ has access to a training data set $\mathcal{D}_i=\big\{\big(\bm{p}^{(\vartheta )}_i, \bm{y}^{(\vartheta )}_i\big)\big\}_{\vartheta =1,\ldots,M_i }$ consisting of $M_i\in\mathbb{N}$ measurement pairs  $(\bm{p}_i^{(\vartheta)},\bm{y}_i^{(\vartheta)}=\boldsymbol{H}(\boldsymbol{q}_i) \ddot{\boldsymbol{q}}_i+\boldsymbol{C}(\boldsymbol{q}_i, \dot{\boldsymbol{q}}_i) \dot{\boldsymbol{q}}_i + \boldsymbol{g}(\boldsymbol{q}_i) - \boldsymbol{u}_i+\bm{\zeta}^{(\vartheta)})$, where $\boldsymbol{y}_i = [{y}_{i1}, \dots,y_{im}]^\top$, $\bm{\zeta}$ is an independent, identical, zero mean Gaussian noise with covariance matrix $\sigma_o^2\bm{I}$. 
\end{assumption}
This assumption allows each agent to have its own independently collected data set without the necessity to share data between them directly. It also necessitates comprehensive measurements of the system states, a common requirement in data-driven control methods, e.g., \cite{umlauftUncertaintyBasedControlLyapunov2018,umlauftFeedbackLinearizationBased2020, Greeff2021}. To address potential measurement noise in GPR, we may transfer noise in the output variable by utilizing Taylor expansion techniques \cite{kim2023model}, or incorporate noise directly into the kernel function \cite{wang2022gaussian}. 
Employing these strategies, the inputs for the GP model can still be treated as effectively noise-free. 

Based on the distributed data sets under \cref{ass_data}, we consider the problem of designing a distributed control law for tracking the virtual leader state $\bm{q}_0$ with the agent states $\bm{q}_i$. Due to the unavailability of the exact dynamics, we cannot expect to achieve exact tracking with asymptotic stability. However, the tracking error of each agent $\bar{\bm{e}}_i \in \mathbb{R}^{2m}$, which is defined as
\begin{equation}
    \bar{\bm{e}}_i(t)= [\bm{e}_i(t)^\top, \dot{\bm{e}}_i(t)^\top]^\top,~\forall i \in \mathcal{V}
\end{equation}
where $\boldsymbol{e}_i = \boldsymbol{q}_i - \boldsymbol{q}_0$, is expected to converge to a small value. This is formalized using the following notion of stability.

\begin{definition}\label{Consensus_tracking}
An ELMAS consisting of $n$ agents achieves consensus tracking if there exists a compact set $\Omega_e \subset \mathbb{R}^{2m}$ containing the origin, so that for all $i \in \mathcal{V}$, $\forall \bar{\bm{e}}_i(0) \in \Omega_e$, there exists a small constant $\psi$ and a finite time $T_e\in\mathbb{R}_+$, such that the tracking error satisfies $\|e_i(t)\|\leq \psi, \forall t\geq T_e$.
\end{definition}

\section{Distributed Learning with GPs}\label{sec_DistributedLearning}
\subsection{Individual Learning}\label{subsec_individualLearning}
A Gaussian Process $\mathcal{GP}(m_{gp}(\bm{p}),k(\bm{p},\bm{p}^{\prime}))$ is a stochastic process where any finite subset of the observations of variables $\{\bm{p}^{(1)}\cdots\bm{p}^{(M)} \}$ is assigned a joint Gaussian distribution defined by a prior mean $m_{gp}(\cdot): \mathbb{R}^{3m} \rightarrow \mathbb{R}$ and a covariance function $k(\cdot): \!\mathbb{R}^{3m} \times \mathbb{R}^{3m} \rightarrow \mathbb{R}_{0,+}$  \cite{rasmussenGaussianProcessesMachine2006}. The prior mean can be used to include approximate models in the regression, and the covariance function reflects structural prior knowledge such as smoothness or periodicity. When no specific structure is known a priori, a frequently used covariance function is the ARD squared exponential kernel $k\left(\bm{p}, \bm{p}^{\prime}\right)=\sigma_{r}^{2} \exp \big(-\frac{1}{2} \sum_{j=1}^{m} {l_{j}}^2(p_{j}-p_{j}^{\prime})^2\big),$
where $\sigma_r\in\mathbb{R}_+$ and $l_j\in\mathbb{R}_+$ are so called hyper-parameters.  

Given a data set satisfying \cref{ass_data}, Gaussian process regression is performed by conditioning the prior GP defined by $m_{gp}(\cdot)$ and $k(\cdot,\cdot)$ on the data set $\mathcal{D}_i$ considering the $i$-th agent with $M_i$ training data pairs. Without loss of generality, we set the prior mean to $0$. Due to the assumption of Gaussian noise, the posterior distribution is again Gaussian. Considering scalar systems, i.e., $m=1$, the posterior has a mean and variance function \cite{rasmussenGaussianProcessesMachine2006}
\begin{align*}
    \mu_i(\bm{p})\! &= \!\bm{k}(\bm{P}_i,\bm{p})^\top (\boldsymbol{K}(\bm{P}_i,\bm{P}_i)+\sigma_o^2\bm{I}_{M_i})^{-1}\bm{Y}_i, \\
    \sigma^2_i(\bm{p}) \! &=\! k(\bm{p},\bm{p})\!-\! \bm{k}(\bm{P}_i,\bm{p})^{\!\top}\!({\bm K}(\bm{P}_i,\bm{P}_i)\!+\!\sigma_o^2\bm{I}_{M_i})^{\!-1}\bm{k}(\bm{P}_i,\bm{p})\!,
\end{align*}
respectively, where 
\begin{equation*}
    \bm{k}(\bm{P}_i,\bm{p}) =  [k(\bm{p}_i^{(1)}, \bm{p}), \dots, k(\bm{p}_i^{(M_i)}, \bm{p})]^\top,
\end{equation*}
the matrix ${\bm K}(\bm{P}_i,\bm{P}_i) = [k(\bm{p}_i^{(a)}, \bm{p}_i^{(b)})]_{a,b = 1,\dots,M_i}$, the training data $\bm{P}_i=[\bm{p}_i^{(1)}\ \ldots\ \bm{p}_i^{(M_i)} ]$, and $\bm{Y}_i=[{y}_i^{(1)}\ \ldots\ {y}_i^{(M_i)} ]^\top$.

In order to apply Gaussian process regression to systems with dimension $m>1$, we consider an independent Gaussian process for each dimension. Under the assumption of equal hyper-parameters for each dimension, the multi-output prediction can then be efficiently computed using
\begin{align}\label{eq:indmean}
    \bm{\mu}_i(\bm{p})&=\bar{\bm k}(\bm{P}_i, \bm{p})^\top \bar{\bm{K}}(\bm{P}_i, \bm{P}_i) \bar{\bm{Y}}_i, \\
    \bm{\sigma}_i^2(\bm{p}) &= \mathbf{1}_m\sigma^2_i(\bm{p}),
\end{align}
where the matrix $\bar{\bm k}(\bm{P}_i, \bm{p}) = \bm{I}_{m}\otimes\bm{k}( \bm{P}_i, \bm{p})$,  $\bar{\bm{K}}(\bm{P}_i, \bm{P}_i)= \bm{I}_{m}\otimes (\boldsymbol{K}(\bm{P}_i,\bm{P}_i)+\sigma_o^2\bm{I}_{M_i})^{-1}$, and training data $\bar{\bm{Y}}_i=[\bm{y}_i^{1}, \dots, \bm{y}_i^{m}]^\top$ with $ \bm{y}_i^{j} = [{y}_{ij}^{(1)},\dots,{y}_{ij}^{(M_i)}],~j=1,2,\dots,m $.

While Gaussian process regression is known to have many beneficial properties for practical usage, it suffers crucially from high computational complexity \cite{rasmussenGaussianProcessesMachine2006}. This is particularly problematic for the posterior variance ${\sigma}^2(\cdot)$, which requires $\mathcal{O}(M_i^2)$ computations for on-line evaluation, even if $(\boldsymbol{K}(\bm{P}_i,\bm{P}_i)+\sigma_o^2\bm{I}_{M_i})^{-1}$ is pre-computed off-line with a complexity of $\mathcal{O}(M_i^3)$. The considerable complexity associated with utilizing posterior variance in control schemes, particularly in distributed learning settings, can often lead to its exclusion from practical usage. As a result, an alternative and more efficient approach is investigated in this paper to achieve cooperative learning objectives without compromising performance. \looseness =-1

\subsection{Cooperative Learning with Correlation-Aware GPs}
As GPR suffers from this inherent computational burden, distributed computing is a promising method. For realizing an effective aggregation, the different predictions' importance must be taken into account by adapting the aggregation weights. This leads to a dependency of the aggregation weights on the posterior variance, which means that each agent $i$ suffers from a $\mathcal{O}(M_i^2)$ for each prediction. In order to address this shortcoming of existing methods, we propose a correlation-aware GP (Cora-GP) algorithm. This approach aggregates the predictions of neighboring agents similarly to existing approaches but employs the prior covariance $\bm{k}(\bm{P},\bm{p})$ between a test point $\bm{p}$ and the training input $\bm{P}$ to determine the aggregation weight of each agent. Therefore, the proposed algorithm sidesteps computing the posterior variance of GPs.

We consider the $i$-th agent with local multi-output GP is trained with the set ${\mathcal{D}}_i$ with $M_i$  pairs of training inputs $\bm{P}_i$ and training outputs $\bm{Y}_i$. Moreover, let $M_{i}^{\min} = \operatorname{min} \{ {M}_{l}| l\in\mathcal{N}_i \}$  be the minimum number of training samples of neighbors of the $i$-th agent. We propose to compute the $j$-th dimension aggregated posterior mean of the multi-output GP of the $i$-th agent
\begin{align}\label{mu_Cora-GP}
\tilde{\mu}_{ij}\left ( \bm{p} \right )  =  \bm{h}_i(\bm{p})^\top \bm{\mu}^{j} \left ( \bm{p} \right ), ~i=1,\dots,n,~j=1\dots,m,
\end{align}
where $\bm{\mu}^{j}(\bm{p}) = [{\mu}_{1j}(\bm{p}),\dots,\mu_{nj}(\bm{p}) ]^\top$. The aggregation weight function $\bm{h}_i(\cdot)\!:\mathbb{R}^{3m}\!\rightarrow \!\mathbb{R}^n$ is defined as
\begin{align}
\label{function_h} {h}_{id}( \bm{p}) &=\begin{cases}
 \frac {{w}_{id}(\bm{p}) }{\sum_{l=1}^{\left | \mathcal{N}_i \right |} {w}_{il}(\bm{p})},  &d \in \mathcal{N}_i\\
0,  &\text{otherwise}
\end{cases},
\end{align}
and the function ${w}(\cdot): \mathbb{R}^{3m} \to \mathbb{R}_+$ is calculated by
\begin{align}
   \label{w_il} w_{il}(\bm{p} )&=
     \frac{\check{a}_{il}(r(t))}{\sigma_{g_i} \sqrt{2\pi}} \exp \Big(- {\Big(\frac{  \| \bm{s}_{il}( \bm{p}) \|}{\tilde{s}_{i}(\bm{p})}-\bar{w}_i \Big)^2}/{2\sigma_{g_i}^2}\Big),
\end{align}
where $l= 1,\dots, n$. The factor $\sigma_{g_i}\in \mathbb{R}_+$ and the parameter
\begin{equation}
\label{bar_w_i}\bar{w}_i\!=\!\operatorname{max} \left\{ \frac{\check{a}_{i1}(r(t))\| \bm{s}_{i1}( \bm{p}  ) \|}{\tilde{s}_{i}(\bm{p})} \cdots \frac{\check{a}_{in}(r(t))\| \bm{s}_{in}( \bm{p}  ) \|}{\tilde{s}_{i}(\bm{p})} \right\},\!
\end{equation}
are the standard deviation and the expected value of the Gaussian distribution \eqref{w_il}, respectively, with $\tilde{s}_{i}(\bm{p}) =  \sum_{l=1}^{\mathcal{N}_i} \check{a}_{il}(r(t))\| \bm{s}_{il}( \bm{p})\|$, which ensures that the maximum value of the correlation function $\bm{s}_{il}(\cdot)$ is the expected value. The correlation function $\bm{s}_{il}(\cdot)$ associated with $\bm{k}(\bm{P}_l,\bm{p} )$ for the $i$-th agent evaluates the correlation between the query point $\bm{p}$ and the training data $\bm{P}_l$. 

Leveraging the results of $\bm{k}(\bm{P}_l,\bm{p} )$, we present the first approach correlation-aware GP with top element (Cora-GP-Top), the function $\bm{s}_{il}(\cdot): \mathbb{R}^{3m} \to \mathbb{R}^{M_i^{\min}}$, considering the training data set $\mathcal{D}_l$ comprises $M_l$ data pairs, is designed as follows
\begin{equation} \label{eqn_Cora-GP_Top}
    \bm{s}_{il}( \bm{p}  ) = \text{Top}( \bm{k}(\bm{P}_l,\bm{p} ), M_i^{\min} ), ~~\text{for} ~~l \in \mathcal{N}_i, 
\end{equation}
otherwise $\bm{s}_{il}(\bm{p}) = \mathbf{0}_{M_i^{\min} \times 1}$. The correlation-aware function $\text{Top}( \bm{k}(\bm{P}_l,\bm{p} ) , M_i^{\min} ): \mathbb{R}^{|\mathcal{D}_l|} \times \mathbb{N} \to \mathbb{R}^{M_i^{\min}}$ selects the first $M_i^{\min}$ largest elements from its input vector, which require $\mathcal{O}(M_l\log(M_l))$ for sorting the values, thereby retaining the most relevant elements based on their magnitudes. 

In order to further dilute the computation time for obtaining the aggregation weights, correlation-aware GP with average elements (Cora-GP-Avg) is developed. In this case, the function $\bm{s}_{il}(\cdot): \mathbb{R}^{3m} \to \mathbb{R}$ simply normalizes the sum of the elements of the vector $\bm{k}(\bm{P}_l,\bm{p} )$ denoted as follows 
\begin{equation} \label{eqn_Cora-GP_Avg}
    \bm{s}_{il}(\bm{p}) = \frac{ \boldsymbol{1}_{n}^\top \bm{k}\big(\bm{P}_l,\bm{p} \big)}{M_l}, ~~\text{for} ~~l \in \mathcal{N}_i, 
\end{equation}
otherwise $\bm{s}_{il}(\bm{p}) = 0$, where this operation only requires $\mathcal{O}(M_l)$ for each agent $l$. Compared to Cora-GP-Top, GoGP-Avg offers the advantage of faster processing as it eliminates the need for sorting values. However, it still effectively captures the underlying correlation relationships.

Within the proposed cooperative learning framework, the posterior mean $\tilde{\mu}_{ij}\left ( \bm{p} \right )$ aggregates the weighted prediction only from the neighbors of agent $i$, which is guaranteed by using the elements $\check{a}_{ij}(r(t))$ of the matrix $\check{\bm{A}}_{r{(t)}}$. Therefore, it does only use information accessible through the communication topology defined by the graph $\mathcal{G}_{r(t)}$. The function $\bm{h}_i(\cdot)$ determines the weights for aggregation with the property $\sum_{d=1}^n h_{id}(\cdot) = 1$. The construction of the weights $h_{id}(\bm{p})$ leads to a dependency on $\bm{k}(\bm{P}_i,\bm{p})$, therefore, it reflects the correlation between inputs and training data. To enhance comprehension of the algorithm's procedure, we furnish a pseudo-code in \cref{alg_cap}. 

The Cora-GP approach has the advantage that its aggregation scheme does not require the posterior variance of individual GPs, but relies solely on $\bm{k}(\bm{P}_i,\bm{p} )$. These vectors are already computed when determining the individual mean functions \eqref{eq:indmean}, such that they come at no additional computational cost. Thereby, the computational complexity in each agent for predictions is reduced to $\mathcal{O}(M_i)$ with Cora-GP-Avg and $\mathcal{O}(M_i\log(M_i))$ with Cora-GP-Top considering the training data set $\mathcal{D}_i$ in contrast to previous works, where this complexity is $\mathcal{O}(M_i^2)$ \cite{deisenrothDistributedGaussianProcesses2015,yangDistributedLearningConsensus2021,ledererCooperativeControlUncertain2023}. Additionally, the functions ${h}_{id}(\cdot)$ in \eqref{function_h} have the beneficial property that their sum equals one. This allows us to derive uniform prediction error bounds under the following additional assumption.
{
\begin{algorithm}[t]
\caption{Cora-GP algorithm}\label{alg_cap}
\begin{algorithmic}
\Require $n \geq 2$  \Comment{number of agents}
\Require OPTION: Choose Cora-GP-Avg or Cora-GP-Top
\If{OPTION \textbf{is} Cora-GP-Top}
    \State Obtain $M_{i}^{min}$
\EndIf
\For{$i =  1 :n$}
\State Calculate $ \bm{k}(\bm{P}_i,\bm{p}) $
\For{$l \in \mathcal{N}_i$}
\State Calculate $ \bm{k}(\bm{P}_l,\bm{p}) $
\State   Update $\bm{s}_{il}(\bm{p})$  based on OPTION via \eqref{eqn_Cora-GP_Top} or \eqref{eqn_Cora-GP_Avg}
\EndFor
\State $\bar{w}_i$, ${w}_{il}(\bm{p} )$ $\leftarrow$  \cref{bar_w_i}, \cref{w_il}
\State  ${h}_{id}( \bm{p})$ $\leftarrow$ \cref{function_h}
\State  $\tilde{\mu}_{ij}\left ( \bm{p} \right )$ $\leftarrow$\cref{mu_Cora-GP}
\EndFor
\end{algorithmic}
\end{algorithm}
}
\begin{assumption} \label{ass_f}
Every component $f_j(\bm{p})$ of the unknown function $ \bm{f}\left ( \bm{p} \right ) $ in \eqref{dyn_agent} with Lipschitz constant $L_{f}$ is a sample obtained from a Gaussian process $ \mathcal{GP}\left ( 0, k\left ( \bm{p},\bm{p}' \right )  \right )  $ with Lipschitz continuous kernel $k\!:\!\mathbb{R}^{3m}\!\times\!\mathbb{R}^{3m}\!\rightarrow\!\mathbb{R}_{0,+}$. 
\end{assumption}
\looseness=-1
This assumption is not restrictive in practice since it merely defines a prior distribution over plausible functions $\bm{f}(\cdot)$ \cite{Lederer2019}. This distribution usually covers a large class of functions, e.g., for squared exponential kernels the support of the distribution corresponds to the continuous functions on a compact set \cite{VanderVaart2011}. Based on this assumption, the following uniform prediction error bound for Cora-GPs can be derived. \looseness=-1

\begin{lemma} \label{lem_bound}
For a compact set for $\bm{p}$ as $  \Omega \in \mathbb{R}^m $, consider the unknown function $\bm{f}(\cdot)$ in \eqref{dyn_agent} satisfying \cref{ass_f} and GPs with the training data set $ {\mathcal{D}}_i$ satisfying \cref{ass_data}, $ \forall i = 1,2,\dots,n $.
Pick $\tau\in\mathbb{R}_+$, $\delta\in(0,1)$ such that $\min_{\bm{p}\in\Omega}\sigma_{ij}^2(\bm{p})\geq \gamma_{ij}^2(\tau)/\varphi (\tau,\delta)$, $\forall j=1,\ldots,m$,
	\begin{align}\label{eq:beta}
	\varphi (\tau,\delta)&=2{m}\log\left( \frac{r_{\Omega}\sqrt{m}}{2\tau} \right)-2\log(\delta),\\
	\gamma_{ij}(\tau)&=(L_f+L_{\mu_{ij}})\tau+\sqrt{\varphi (\tau,\delta) L_{\sigma_{ij}^2}\tau},
	\end{align}
for $r_{\Omega}=\max_{\bm{p},\bm{p}'\in\Omega}\|\bm{p}-\bm{p}'\|$, $L_{\mu_{ij}}$ and $L_{\sigma_{ij}^2}$ the Lipschitz constants of the individual GP mean and variance functions, respectively. Then, with probability of at least $(1-\delta)^m$, it holds 
for the proposed Cora-GP method in \eqref{mu_Cora-GP} 
that
$\|\bm{f}-\tilde{\bm{\mu}}_i\| \leq \tilde{\eta}_{i}(\bm{p},\delta)$ for all $i=1,\ldots,m$, where $\tilde{\bm{\mu}}_{i} := [ \tilde{{\mu}}_{i1},\dots,  \tilde{{\mu}}_{im}]^\top$ and $\tilde{\eta}_{i}(\bm{p},\delta) =  \| [{\eta}_{i1}(\bm{p},\delta),\dots,{\eta}_{im}(\bm{p},\delta) ]\|$   with
\begin{align}\label{eq:un_bound}
		\eta_{ij}(\bm{p},\delta)= 2\sqrt{\varphi \left(\tau,{\delta}/{n}\right)}\bm{h}_i( \bm{p})^\top \bm{\sigma}^{j}(\bm{p}),
\end{align}
where $j =1,\dots, m$, $\bm{\sigma}^{j}(\bm{p}) = [{\sigma}_{1j}(\bm{p}),\dots,{\sigma}_{nj}(\bm{p})]^\top$.
\end{lemma}
\begin{proof}
According to the Cora-GP algorithm \eqref{mu_Cora-GP}, it is trivial to show the $j$-th dimensional unknown $\bm{f}(\cdot)$ prediction error $|\Delta f_{ij}(\bm{p})|:=|f_j(\bm{p})-\tilde{\mu}_{ij}(\bm{p})|$ with the property in \eqref{function_h} is bounded by
\begin{align}
	&|\Delta f_{ij}(\bm{p})|= \left| \bm{h}_i(\bm{p})^\top \bm{\mu}^{j} (\bm{p}) - \bm{h}_i(\bm{p})^\top \mathbf{1}_n f_j(\bm{p}) \right|, \\
 &=\left| \bm{h}_i(\bm{p})^\top \big (\bm{\mu}^{j} (\bm{p}) - \mathbf{1}_n f_j(\bm{p} ) \big) \right|, \nonumber \\
	&\leq\bm{h}_i(\bm{p})^\top \big[|{\mu}_{1j} (\bm{p}) - f_j(\bm{p} )|,\dots, |{\mu}_{nj} (\bm{p}) - f_j(\bm{p} )| \big]^\top.\nonumber
\end{align}
Similarly to \cite{yangDistributedLearningConsensus2021}, we have the joint $j$-th prediction error of the $i$-th agent with the probability of at least $1-\delta$ bounded with
\begin{equation} \label{eq:beta_tau_bound}
    |\Delta f_{ij}(\bm{p})| \leq \bm{h}_i(\bm{p})^\top \big( \sqrt{\varphi \left(\tau,{\delta}/{n}\right)}\bm{\sigma}^{j}(\bm{p})\!+\!\bm{\gamma}^{j}(\tau)  \big),
\end{equation}
for $\tau\in\mathbb{R}_+$ and $\varphi (\tau,\delta)=2\log(M(\tau,\Omega)/\delta)$, where $\bm{\gamma}^{j}(\tau) = [\gamma_{1j}(\tau),\dots,\gamma_{nj}(\tau)]^\top$,  $M(\tau,\Omega)$ denotes the $\tau$-covering number of $\Omega$. By overapproximating $\Omega$ through a hypercube with edge length $r_{\Omega}$, the covering number $M(\tau,\Omega)$ can be bounded by $(r_{\Omega}\sqrt{m}/(2\tau))^{d}$, which yields identity \eqref{eq:beta}. Moreover, as $M_i$ is finite for all $i=1,\ldots,n$ and the training targets $\bm{y}_i^{(\cdot)}$ are perturbed by Gaussian noise, the posterior standard deviation is positive, i.e., there exists a ${\sigma}_{\min} \in \mathbb{R}_{0,+}$ such that $\sigma_{ij}(\bm{p})\geq {\sigma}_{\min}$ for all $\bm{p}\in\Omega$, $i=1,\ldots,n, ~j=1,\ldots,m$. Moreover, $\varphi (\tau,\delta)$ is monotonically decreasing in $\tau$, while $\gamma(\tau)$ is monotonically growing. Therefore, there exists a $\tau$ such that $\min_{\bm{p}\in\Omega}\sigma_{ij}^2(\bm{p})\geq{\sigma}_{\min}^2\geq  \gamma_{ij}^2(\tau)/\varphi (\tau,\delta)$ for all $i=1,\ldots,n, ~j=1,\ldots,m$, which, together with \eqref{eq:un_bound}, allows us to simplify \eqref{eq:beta_tau_bound} to $|\Delta f_{ij}(\bm{p})|\leq  \eta_{i,j}(\bm{p},\delta)$. 
Similar to \cite{umlauftUncertaintyBasedControlLyapunov2018}, using the fact that $  \bigcap_{j=1}^{m}  |\Delta f_{ij}(\bm{p})|\leq \eta_{i,j}(\bm{p},\delta)$, then
\begin{equation}
    \|\Delta\bm{f}_i\| \leq  \| [{\eta}_{i1}(\bm{p},\delta),\dots,{\eta}_{im}(\bm{p},\delta) ]  \|
\end{equation}
yields the result with the probability of at least $(1-\delta)^m$.
\end{proof}
This lemma establishes a uniform prediction error bound over the compact domain $\Omega$. For the derivation of the Lipschitz constants $L_{\mu_{ij}}$ and $L_{\sigma_{ij}^2}$, we refer to \cite{Lederer2019}. 

\section{Consensus Tracking under Switching Topology }\label{sec_ConsensusTracking}
To achieve consensus tracking with the ELMAS \eqref{dyn_agent} under switching semi-Markov topologies, for each agent $i~(i \in \mathcal{V})$, we let $\hat{\bm{f}}_i(\boldsymbol {p_i}) = \tilde{\boldsymbol{\mu}}_i(\boldsymbol {p_i}):=[\tilde{\mu}_{i1}(\boldsymbol {p_i}),\dots,\tilde{\mu}_{im}(\boldsymbol {p_i})]^\top$  be the prediction of the unknown dynamics $\bm{f}(\cdot)$ from agent $i$ obtained by a Cora-GP \eqref{mu_Cora-GP} and define a feedback linearizing distributed control law of the form 
\begin{equation} \label{controller2}
    \boldsymbol{u}_i =  {c}_i\boldsymbol{H}(\boldsymbol{q_i}) \boldsymbol{\nu}_{i}  + \boldsymbol{C}(\boldsymbol{q}_i, \dot{\boldsymbol{q}_i}) \dot{\boldsymbol{q}_i}+\boldsymbol{g}(\boldsymbol{q}_i)- \hat{\bm{f}}(\boldsymbol {p_i}),
\end{equation}
where $\boldsymbol{\nu}_{i}$ is the synchronization error defined by 
\begin{align}\label{eq:conslaw}
   \boldsymbol{\nu}_{i}&=  -\sum_{j=0}^n a_{ij}(r(t))[\alpha (\boldsymbol{q}_{i}-\boldsymbol{q}_{j}) + (\dot{\boldsymbol{q}_{i}}-\dot{\boldsymbol{q}_{j}}) ] \nonumber \\
   &= \alpha \Delta \bm{q}_{i} + \Delta \dot{\bm{q}}_{i},
\end{align}
where $ \Delta \bm{q}_{i}:= \sum_{j=0}^n a_{i j}(r(t))(\boldsymbol{q}_{j}-\boldsymbol{q}_{i})$ is the consensus tracking error corresponding to the $i$-th agent and its derivative $ \Delta \dot{\bm{q}}_{i}:=\sum_{j=0}^n a_{ij}(r(t))(\dot{\boldsymbol{q}}_{j}-\dot{\boldsymbol{q}}_{i})$.  The constant coefficient $\alpha \in \mathbb{R}_+$ is set to be positive.

Since the Cora-GP predictions \eqref{mu_Cora-GP} and the synchronization error \eqref{eq:conslaw} rely only on locally available information or values accessible via the communication network, the control law \eqref{controller2} can be implemented in a distributed fashion. Moreover, due to the error bounds for Cora-GP predictions in Lemma~\ref{lem_bound}, we can guarantee that \eqref{controller2} achieves consensus tracking. This is shown in the following theorem.
\begin{theorem}\label{theorem1}
Consider an ELMAS consisting of $n$ agents described by \eqref{dyn_agent} and a virtual leader described by \eqref{dyn_laeder} under \cref{ass_leader} with switching topologies $\bar{\mathcal{G}}_{r(t)}$ satisfying \cref{ass_switchinggraph}, where $r(t)$ is governed by a semi-Markov process with the finite state space $\mathcal{P}=\{1,2,\ldots,N \}$ and jump times $\tau$ following the distribution functions $F_{ij}(t)$. By using the proposed distributed learning control law \eqref{controller2} with the chosen gain $c_i\in\mathbb{R}_{+}$, employing Cora-GP algorithm based on the agent data sets $\mathcal{D}_i$, $i=1,\dots,n$ satisfying \cref{ass_f}, the ELMAS achieves consensus tracking with probability $(1-\delta)^m$, $\delta \in (0, 1)$, if  
\begin{equation} \label{eq_Upsilon}
    \bm{\Phi}_1 \!=\! \begin{bmatrix}
  \min_{ r(t) \in\mathcal{P}} \underline{\sigma}(\bm{c}\tilde{\mathbf{L}}_{r(t)} \!-\! \alpha\bm{I}_{nm})& -\frac{(1+{\alpha }^2  )}{2}\\
  -\frac{(1+{\alpha }^2  )}{2}& {\alpha}
\end{bmatrix} \!\succ \!0,\!
\end{equation}
is satisfied, where  $\bm{c}= \operatorname{diag}(c_1\bm{I}_m,\dots,c_n\bm{I}_m)$ and $\tilde{\mathbf{L}}_{r(t)}\! =\! \tilde{\bm{L}}_{r(t)} \otimes \bm{I}_m$. Then the consensus tracking of the ELMAS achieves with the overall consensus tracking error 
\begin{equation}
    \bar{\boldsymbol{e}} =  [ {\bm{e}}_1^\top,{\bm{e}}_2^\top, ~\cdots ~ {\bm{e}}_n^\top, \dot{\bm{e}}_1^\top, \dot{\bm{e}}_2^\top, ~\cdots ~ \dot{\bm{e}}_n^\top  ]^\top
\end{equation}
bounded with
\begin{align} \label{trackingerror_bound}
    \| \bar{\boldsymbol{e}} \|\leq \frac{(1+\alpha)\| \bm{\Phi}_2 \|}{2 \min_{r(t)} \underline{\sigma}(\tilde{\boldsymbol{L}}_{r(t)}) \sqrt{ (\underline{\sigma}(\bm{\Phi}_1) + 1/2) }},
\end{align}
where 
\begin{equation}\label{eq_Upsilon2}
    \bm{\Phi}_2 \!= \!\Big[ \max_{r\in\mathcal{P}} \bar{\sigma}(\tilde{\boldsymbol{L}}_{r(t)}) \Big( \frac{\tilde{\eta}(\bm{p},\delta)}{\min_{\bm{q}\in \mathbb{X}}\bar{\sigma}(\bm{H}(\bm{q})) } + \sqrt{n}\bar{f}_r \Big) ~ 0 \Big].\!
\end{equation}

\end{theorem}

\begin{proof}
Before delving into the analysis of system stability, we first present the collective dynamics of the ELMAS \eqref{dyn_agent} 
\begin{equation}
    \mathbf{H}(\mathbf{q})\ddot{\mathbf{q}} + \mathbf{C}({\mathbf{q}},\dot{\mathbf{q}})\dot{\mathbf{q}} + \mathbf{g}(\mathbf{q}) +  {\mathbf{f}}(\mathbf{p}) = \mathbf{u},
\end{equation}
where  $\mathbf{q} = [\bm{q}_1^\top,\dots,\bm{q}_n^\top]^\top$, ${\mathbf{f}}(\mathbf{p})= [{\bm{f}}(\boldsymbol {p_1})^\top,\dots,{\bm{f}}(\boldsymbol {p_n})^\top]^\top$, $\mathbf{p}\! =\![\bm{p}_1^\top,\dots,\bm{p}_n^\top]^\top$, $\mathbf{u} \!=\![\boldsymbol{u}_1^\top,\dots,\boldsymbol{u}_n^\top]^\top$, 
\begin{align*}
    \mathbf{C}({\mathbf{q}},\dot{\mathbf{q}}) &= \mathrm{blkdiag}\big(\boldsymbol{C}(\boldsymbol{q}_1, \dot{\boldsymbol{q}}_1),\dots, \boldsymbol{C}(\boldsymbol{q}_n, \dot{\boldsymbol{q}}_n) \big),\\
    \mathbf{H}(\mathbf{q}) &=\mathrm{blkdiag}(\boldsymbol{H}(\boldsymbol{q}_1)\cdots \boldsymbol{H}(\boldsymbol{q}_n)  ),\\
    \mathbf{g}(\mathbf{q}) &=[\boldsymbol{g}_1(\boldsymbol{q}_1)^{\top} \cdots\boldsymbol{g}_n(\boldsymbol{q}_n)^{\top}]^\top.
\end{align*}
Considering the distributed controller \eqref{controller2}, the collective control law as follows
\begin{equation}\label{controller_MAS}
    \mathbf{u} = \bm{c}\mathbf{H}(\mathbf{q})\bm{\nu} + \mathbf{C}(\mathbf{q},\dot{\mathbf{q}})\dot{\mathbf{q}} + \mathbf{g}(\mathbf{q}) -  \hat{\mathbf{f}}(\mathbf{p}),
\end{equation}
where $\hat{\mathbf{f}}(\mathbf{p})= [\hat{\bm{f}}(\boldsymbol {p_1})^\top,\dots,\hat{\bm{f}}(\boldsymbol {p_n})^\top]^\top,$  $\bm{\nu} \!=\! [\boldsymbol{\nu}_{1}^\top\!,\dots,\!\boldsymbol{\nu}_{n}^\top]^\top$.

We consider the Lyapunov candidate 
\begin{equation}
    V = \frac{1}{2} \bm{\nu}^\top\bm{\nu} +\frac{1}{2} \Delta\mathbf{q}^\top \Delta\mathbf{q},
\end{equation}
where $\Delta { \mathbf{q} } = [\Delta { \bm{q}}_1^\top,\dots,\Delta { \bm{q}}_n^\top]^\top$.
Considering the definition of synchronization error \eqref{eq:conslaw}, the derivative of $\bm{\nu}$ is
\begin{align}\label{eq_dotnu}
\dot{\bm{\nu}} = &~ \alpha \Delta\dot{\mathbf{q}} \\
    &+\tilde{\mathbf{L}}_{r(t)}\Big[\mathbf{H}(\mathbf{q})^{-1}(\mathbf{C}(\mathbf{q},\dot{\mathbf{q}})\dot{\mathbf{q}}+ \mathbf{g}(\mathbf{q}) - \mathbf{f}(\mathbf{p}) - \mathbf{u}) + \ddot{\mathbf{q}}_l \Big ], \nonumber
\end{align}
where $\ddot{\mathbf{q}}_l = \bm{I}_n\otimes\ddot{\bm{q}}_l$. Similarly, the derivative of $\Delta{\mathbf{q}}$ is derived as
\begin{equation}\label{eq_dotDeltaq}
    \Delta\dot{\mathbf{q}} = \bm{\nu} - ({\alpha}\bm{I}_{nm})\Delta\mathbf{q}
\end{equation}
Combining \eqref{eq_dotnu} and \eqref{eq_dotDeltaq}, one has
\begin{align}\label{Vdot}
    &\dot{V} = \bm{\nu}^\top \dot{\bm{\nu}} +\Delta\mathbf{q}^\top\Delta\dot{\mathbf{q}},\\
    &=(\alpha\bm{\nu}^\top + \Delta\mathbf{q}^\top)  (\bm{\nu} - ({\alpha}\bm{I}_{nm})\Delta\mathbf{q}) \nonumber\\
    &+ \bm{\nu}^\top \tilde{\mathbf{L}}_{r(t)} \Big[\mathbf{H}(\mathbf{q})^{-1}(\mathbf{C}(\mathbf{q},\dot{\mathbf{q}})\dot{\mathbf{q}} + \mathbf{g}(\mathbf{q}) - \mathbf{f}(\mathbf{p}) - \mathbf{u}) - \ddot{\mathbf{q}}_l \Big ]. \nonumber
\end{align}
Substituting the control law \eqref{controller_MAS} into \eqref{Vdot}, we have
\begin{align*}
    \dot{V} =& -\bm{\nu}^\top (\bm{c}\tilde{\mathbf{L}}_{r(t)} \!-\! {\alpha}\bm{I}_{nm} )\bm{\nu}\! -\! {\alpha}\Delta\mathbf{q}^\top \Delta\mathbf{q} - {\alpha }^2\bm{\nu}^\top \Delta\mathbf{q}\nonumber \\
    & +\!\Delta\mathbf{q}^\top \bm{\nu}\!-\!\bm{\nu}^\top \tilde{\mathbf{L}}_{r(t)}(\mathbf{H}(\mathbf{q})^{-1}\!\Delta\mathbf{f}(\mathbf{q}) \!-\! \ddot{\mathbf{q}}_l).
\end{align*}
 By employing Lemma \ref{lem_bound}, this expression can be bounded by
\begin{align}\label{Vdot3}
   &\dot{V}  \leq -\underline{\sigma}(\bm{c}\tilde{\mathbf{L}}_{r(t)} - \alpha\bm{I}_{nm} ) \|\bm{\nu}\|^2  \!+\! (1+{\alpha }^2 )\|\bm{\nu} \| \|\Delta\mathbf{q}\|  \\
    &-\! {\alpha}\|\Delta\mathbf{q}\|^2+ \bar{\sigma}(\tilde{\bm{L}}_{r(t)} ) \big(\bar{\sigma}(\mathbf{H}(\mathbf{q})^{-1}) \tilde{\eta}(\bm{p},\delta)  + \sqrt{n}\bar{f}_r \big) \|\bm{\nu}\|,\nonumber
\end{align}
where $\tilde{\eta}(\bm{p},\delta)\! =  \|[\tilde{\eta}_{1}(\bm{p}_1,\delta),\dots,\tilde{\eta}_{n}(\bm{p}_n,\delta) ]^\top\|$. By leveraging the fact that $\bar{\sigma}(H^{-1}) = 1/\underline{\sigma}(H)$, we have the reduced form of \eqref{Vdot3} as follows
\begin{align} \label{Vdot_2}
    \dot{V} \leq  -\bm{\rho}^\top \bm{\Phi}_1 \bm{\rho}+ \bm{\Phi}_2 \bm{\rho},
\end{align}
where $ \bm{\rho}=[ \| \bm{\nu} \|,  \|\Delta \mathbf{q} \|  ]^\top $,  and $\bm{\Phi }_1$ and $\bm{\Phi }_2$ are defined in \eqref{eq_Upsilon} and \eqref{eq_Upsilon2}, respectively. By using the Young's inequality, we have
\begin{equation}\label{Young_inequ}
    \| \bm{\Phi}_2 \| \|\bm{\rho} \| \leq \frac{\| \bm{\Phi}_2 \|^2+\|  \bm{\rho}\|^2}{2}.
\end{equation}
Substituting \eqref{Young_inequ} into \eqref{Vdot_2}, one has
\begin{equation}\label{Vdot_3}
    \dot{V} \leq -(\underline{\sigma}(\bm{\Phi}_1) + \frac{1}{2})\| \bm{\rho}\|^2 + \frac{\| \bm{\Phi}_2 \|^2}{2}.
\end{equation}
Furthermore, from \eqref{Vdot_3} follows the fact \cite[Lemma 1]{zeng-guanghouDecentralizedRobustAdaptive2009} that 
\begin{align}
{V}(t) \leq & V(0) \exp (-\tilde{\Phi}_1 t) + \frac{\| \bm{\Phi}_2 \|^2 }{2\tilde{\Phi}_1}\big(1- \exp (-\tilde{\Phi}_1t)\big),
\end{align}
where $\tilde{\Phi}_1=\underline{\sigma}(\bm{\Phi}_1) + {1}/{2}$ and $V(0)$ indicates the value of $V$ at time $t=0$. Therefore, there exist a bounded $ \upsilon \in \mathbb{R}_+$ and $T_e>0$ such that $\{ V(t)| \left | V(t) \right | \leq \upsilon \}$ when $t\geq T_e$. 
Using the facts that 
\begin{equation}
     \| [ \Delta\mathbf{q}, \Delta\dot{\mathbf{q}}  ] \| \leq  (1+ \alpha ) \| \bm{\rho} \|
\end{equation}
due to \eqref{eq_dotDeltaq} and the overall consensus tracking error $ \bar{\boldsymbol{e}}$ satisfies
\begin{equation}
      \big[ \Delta\mathbf{q}^\top ~ \Delta\dot{\mathbf{q}}^\top  \big]^\top   = -\big(\boldsymbol{I}_2 \otimes \tilde{\mathbf{L}}_{r(t)}  \big) \bar{\boldsymbol{e}},
\end{equation}
we have
\begin{align}
    \| \bar{\boldsymbol{e}}\| \leq \frac{(1+\alpha)\| \bm{\Phi}_2 \|}{2  \min_{r(t)} \underline{\sigma}(\tilde{\boldsymbol{L}}_{r(t)}) \sqrt{ (\underline{\sigma}(\bm{\Phi}_1) + 1/2) }},
\end{align}
is bounded, which concludes the proof.  
\end{proof}
From Theorem 1, we can derive several insightful properties. One crucial observation is that the high connectivity of the graph $\bar{\mathcal{G}}_{r(t)}$, characterized by relatively large singular values of $\tilde{\boldsymbol{L}}_{r(t)}$ at state $r(t)$, leads to a diminished tracking error. Therefore, the guaranteed tracking error bound is dictated by the graph's lowest level of connectivity observed during the process. Additionally, it is noteworthy that the matrix $\bm{\Phi}_1$ associated with the control gains also appears in the denominator of \eqref{trackingerror_bound}. This implies that by appropriately designing the values of $c_i$ and $\alpha$, it is possible to achieve arbitrarily small ultimate tracking error bounds.

\section{Simulation}\label{sec_Simulation}

For demonstrating the effectiveness of the proposed control law \eqref{controller2} based on Cora-GP predictions \eqref{mu_Cora-GP}, 
we consider 4 homogeneous 2-link robotic manipulators as described in \cite{murrayMathematicalIntroductionRobotic2017}.  It is assumed that the 4 manipulators possess the same parameters, in particular, point masses for the links $m_1=m_2=1$ kg, length of the links $l_1 = l_2 = 1 $ m. The unknown dynamics $\bm{f}(\cdot)$ is chosen as $${\boldsymbol{f}}(\boldsymbol{p}) =  
[q_{2}\sin(4q_{2})\!+\!\cos(q_{1}),
q_{2}\sin(0.2q_{1}^2) \!+\! \cos(q_{1})
]^T.$$ A total of $1150$ training samples is collected by sampling $\bm{q}$ from the domain $[-1,1]^2$ and perturbing the resulting output measurements $\bm{f}(\bm{p})$ by zero-mean Gaussian noise with $\sigma_0=0.1$. The samples of the state $\bm{q}$ are unevenly distributed among the agents as shown in \cref{fig_data}, leading to $M_1 = 350,~M_2 =250,~M_3 = 300,~M_4 = 250$. Moreover, the trajectory of the virtual leader is chosen as ${q}_{0 1}(t)=0.8\cos (0.02 \pi t),~{q}_{0 2}(t)=-0.8\sin (0.02 \pi t)$. The control gains are set to $\alpha =2$ and $c_i = 2$, respectively, and the factor to $\sigma_{g_i} = 0.15$, $\forall i \in \mathcal{V}$. We use a set of $6$ switching topologies, which are shown in \cref{fig_switch}. In particular, we choose ${F}_{i j}(\tau) \sim e^{-0.5 \tau}$. The initial states of each dimension of $\boldsymbol{q}$ and $\dot{\boldsymbol{q}}$ are randomly uniform distribution in the interval $[0,1.6]$ and $[-0.8, 0.8]$, respectively, for each agent. The initial probability of entering the six states is given randomly as $[0.1705 \quad  0.2073 \quad 0.0045 \quad 0.1863 \quad  0.2456 \quad 0.1858]$. 
The transition probability matrix $\mathbf{P}_r$ is randomly chosen and the switching signals satisfy \cref{ass_irr}.
\begin{align*}
\mathbf{P}_r =  \begin{bmatrix}
0.00 & 0.10 & 0.20 & 0.10 & 0.50 & 0.10 \\
0.22 & 0.00 & 0.28 & 0.06 & 0.28 & 0.17 \\
0.07 & 0.14 & 0.00 & 0.36 & 0.07 & 0.36 \\
0.15 & 0.23 & 0.31 & 0.00 & 0.08 & 0.23 \\
0.08 & 0.23 & 0.23 & 0.15 & 0.00 & 0.31 \\
0.06 & 0.23 & 0.18 & 0.24 & 0.29 & 0.00
\end{bmatrix},
\end{align*}
which determines the probability of one state entering another state. 

\begin{figure}[t]
\vspace{4pt}
\centering  
\begin{subfigure}[b]{0.55\columnwidth}
\input{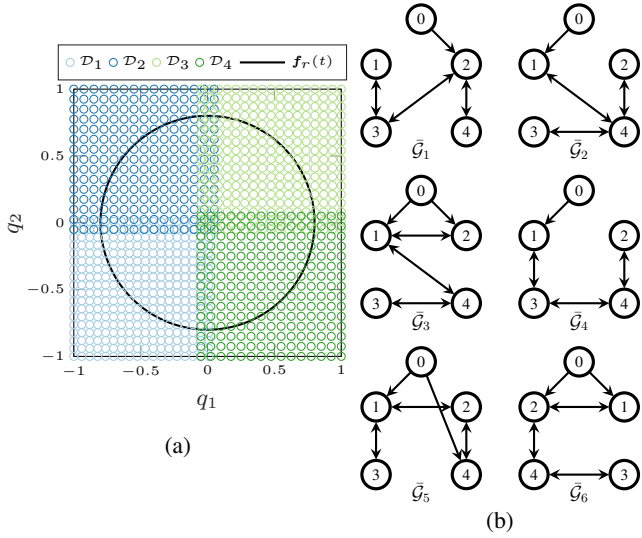}
\caption{}
\label{fig_data}
\vspace{1.0cm}
\end{subfigure}\hfill
\begin{subfigure}[b]{0.43\columnwidth}
    \definecolor{mycolor1}{RGB}{31,120,180}
\definecolor{mycolor2}{RGB}{208,28,139}

\begin{tikzpicture}[scale=0.6,every node/.append style={transform shape}]
\tikzstyle{agent} = [very thick, circle, minimum width = 0.5cm, minimum height=0.5cm,text centered, draw = black ,fill= white ]
\tikzstyle{arrow1} = [thick, <->,>=stealth]
\tikzstyle{arrow} = [thick, ->,>=stealth]
\node (leader1) [agent, xshift=-3.5cm,yshift=1cm] {\normalsize
 0};
\node (agent1) [agent, xshift=-4.5cm] {\normalsize
 1};
\node (agent2) [agent, xshift=-2.5cm] {\normalsize
 2};
\node (agent3) [agent, yshift=-1.5cm, xshift=-4.5cm] {\normalsize
  3};
\node (agent4) [agent, yshift=-1.5cm, xshift=-2.5cm] {\normalsize
 4};
\node (leader2)[agent, yshift=1cm] {\normalsize
 0};
\node (agent5) [agent, xshift=-1cm] {\normalsize
 1};
\node (agent6) [agent, xshift=1cm] {\normalsize
 2};
\node (agent7) [agent, yshift=-1.5cm, xshift=-1cm] {\normalsize
  3};
\node (agent8) [agent, yshift=-1.5cm, xshift=1cm] {\normalsize
 4};
\node (leader3)[agent, xshift=-3.5cm, yshift=-2.8cm] {\normalsize
 0};
\node (agent9) [agent, xshift=-4.5cm, yshift=-3.8cm] {\normalsize
 1};
\node (agent10) [agent, xshift=-2.5cm, yshift=-3.8cm] {\normalsize
 2};
\node (agent11) [agent, yshift=-5.3cm, xshift=-4.5cm] {\normalsize
  3};
\node (agent12) [agent, yshift=-5.3cm, xshift=-2.5cm] {\normalsize
 4};
\node (leader4) [agent, yshift=-2.8cm] {\normalsize
 0};
\node (agent13) [agent, xshift=-1cm,yshift=-3.8cm,] {\normalsize
 1};
\node (agent14) [agent, xshift=1cm,yshift=-3.8cm,] {\normalsize
 2};
\node (agent15) [agent, yshift=-5.3cm, xshift=-1cm] {\normalsize
  3};
\node (agent16) [agent, yshift=-5.3cm, xshift=1cm] {\normalsize
 4};
\node (leader5) [agent, xshift=-3.5cm, yshift=-6.6cm] {\normalsize
 0};
\node (agent17) [agent, xshift=-4.5cm,yshift=-7.6cm,] {\normalsize
 1};
\node (agent18) [agent, xshift=-2.5cm,yshift=-7.6cm,] {\normalsize
 2};
\node (agent19) [agent, yshift=-9.1cm, xshift=-4.5cm] {\normalsize
  3};
\node (agent20) [agent, yshift=-9.1cm, xshift=-2.5cm] {\normalsize
 4};
\node (leader6) [agent, yshift=-6.6cm] {\normalsize
 0};
\node (agent21) [agent, xshift=1cm,yshift=-7.6cm,] {\normalsize
 1};
\node (agent22) [agent, xshift=-1cm,yshift=-7.6cm,] {\normalsize
 2};
\node (agent23) [agent, yshift=-9.1cm, xshift=1cm] {\normalsize
  3};
\node (agent24) [agent, yshift=-9.1cm, xshift=-1cm] {\normalsize
 4};
\draw [arrow]    (leader1) to (agent2);
\draw [arrow1]   (agent1) to  (agent3);
\draw [arrow1]   (agent2) to  (agent3);
\draw [arrow1]   (agent2) to  (agent4);
\draw [arrow]    (leader2) to (agent5);
\draw [arrow1] (agent6) to  (agent8);
\draw [arrow1] (agent7) to  (agent8);
\draw [arrow1] (agent5) to  (agent8);

\draw [arrow]    (leader3) to (agent10);
\draw [arrow]    (leader3) to (agent9);
\draw [arrow1] (agent9) to  (agent10);
\draw [arrow1] (agent9) to  (agent12);
\draw [arrow1] (agent11) to  (agent12);

\draw [arrow]    (leader4) to (agent13);
\draw [arrow1] (agent13) to  (agent15);
\draw [arrow1] (agent14) to  (agent16);
\draw [arrow1] (agent15) to  (agent16);
\draw [arrow]    (leader5) to (agent17);
\draw [arrow]    (leader5) to (agent20);
\draw [arrow1] (agent17) to  (agent18);
\draw [arrow1] (agent17) to  (agent19);
\draw [arrow1] (agent18) to  (agent20);
\draw [arrow]    (leader6) to (agent21);
\draw [arrow]    (leader6) to (agent22);
\draw [arrow1] (agent21) to  (agent22);
\draw [arrow1] (agent23) to  (agent24);
\draw [arrow1] (agent22) to  (agent24);
\draw node   [xshift=-3.5cm,yshift=-1.9cm] 
{\large$\bar{{\mathcal{G}}}_1$};
\draw node   [yshift=-1.9cm]{\large$\bar{{\mathcal{G}}}_2$};
\draw node   [xshift=-3.5cm,yshift=-5.7cm] 
{\large$\bar{{\mathcal{G}}}_3$};
\draw node   [yshift=-5.7cm] 
{\large$\bar{{\mathcal{G}}}_4$};
\draw node   [xshift=-3.5cm,yshift=-9.5cm] 
{\large$\bar{{\mathcal{G}}}_5$};
\draw node   [yshift=-9.5cm] 
{\large$\bar{{\mathcal{G}}}_6$};

\end{tikzpicture}
    \vspace{-0.6cm}
    \caption{}
    \label{fig_switch}
\end{subfigure}
\vspace{-0.1cm}
\caption{(a) Switching communication topologies; (b) The agent $i$ has the data set $\mathcal{D}_i$. Meanwhile, the leader's trajectory $\bm{f}_r(t)$ crosses the whole data area. }
\end{figure}

\begin{figure}[t]
\vspace{4pt}
\centering
\input{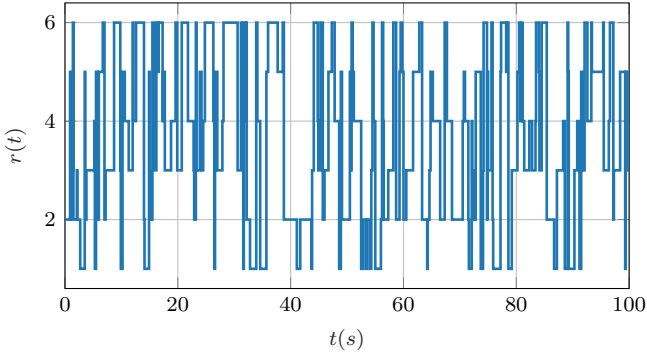}
\vspace{-0.8cm}
\caption{Switching states.}
\label{signals}
\end{figure}

In order to demonstrate the high control performance and a reduction in computational complexity, we compare the proposed control law \eqref{controller2} to the same control law with cooperative GPs (CGP) as proposed in \cite{yangDistributedLearningConsensus2021}, and individual GP (IGP) meaning each agent estimates the uncertainties with its own prediction independently. The corresponding trajectory of the norm of the overall tracking errors $\bar{\bm{e}}$ of one trail over the whole simulation time is illustrated in \cref{fig_errors}. The disparity in errors becomes evident when comparing the scenario without GP to the others. Notably, CGP, Cora-GP-Top, and Cora-GP-Avg exhibit substantial reductions in errors compared to the IGP method. Moreover, the Mont-Carlo test with $100$ times from different initial conditions demonstrates that both Cora-GP-Top and Cora-GP-Avg achieve equivalent predictive accuracy of CGP while circumventing the computational overhead associated with calculating the posterior variance of GP in \cref{fig_MCtest}. \looseness=-1
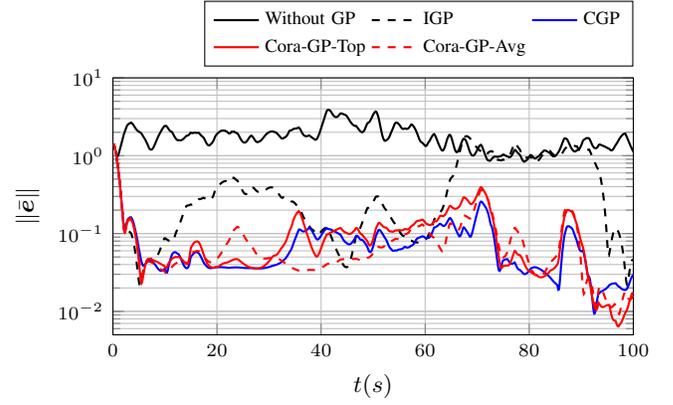
\begin{figure}[t] 
    \centering
    \def\file{tikz/Errors.txt}
    \begin{tikzpicture}
        \begin{axis}[xlabel={$t(s)$},ylabel={$\| \bar{\bm{e}} \|$}, xmin=0, xmax = 100, ymin=0.005, ymax=10,ymode=log,yminorticks=true,xmajorgrids, ymajorgrids, yminorgrids,legend columns=3,
            width=8.5cm,height=5cm,legend style={at={(1,1.3)}},
            clip mode=individual]
            \addplot[black, thick] table[x = timeLogger , y  = NormErrorX_allLogger ]{\file};
            \addplot[black, thick, dashed]    table[x = timeLogger , y  = NormErrorX_allIndivLogger ]{\file};
            \addplot[blue, thick]    table[x = timeLogger , y  = NormErrorX_allCoopLogger ]{\file};
            \addplot[red, thick]    table[x = timeLogger , y  = NormErrorX_allTopLogger ]{\file};
            \addplot[red, thick, dashed]    table[x = timeLogger , y  = NormErrorX_allAvgLogger ]{\file};
            \legend{Without GP$~$, IGP, CGP, Cora-GP-Top, Cora-GP-Avg}
        \end{axis}
    \end{tikzpicture}
    \vspace{-0.6cm}
    \caption{
        Tracking error plots for the simulated scenarios.
    }
    \label{fig_errors}
\end{figure}
\looseness=-1
\begin{figure}
    \centering
    \includegraphics[scale=0.62]{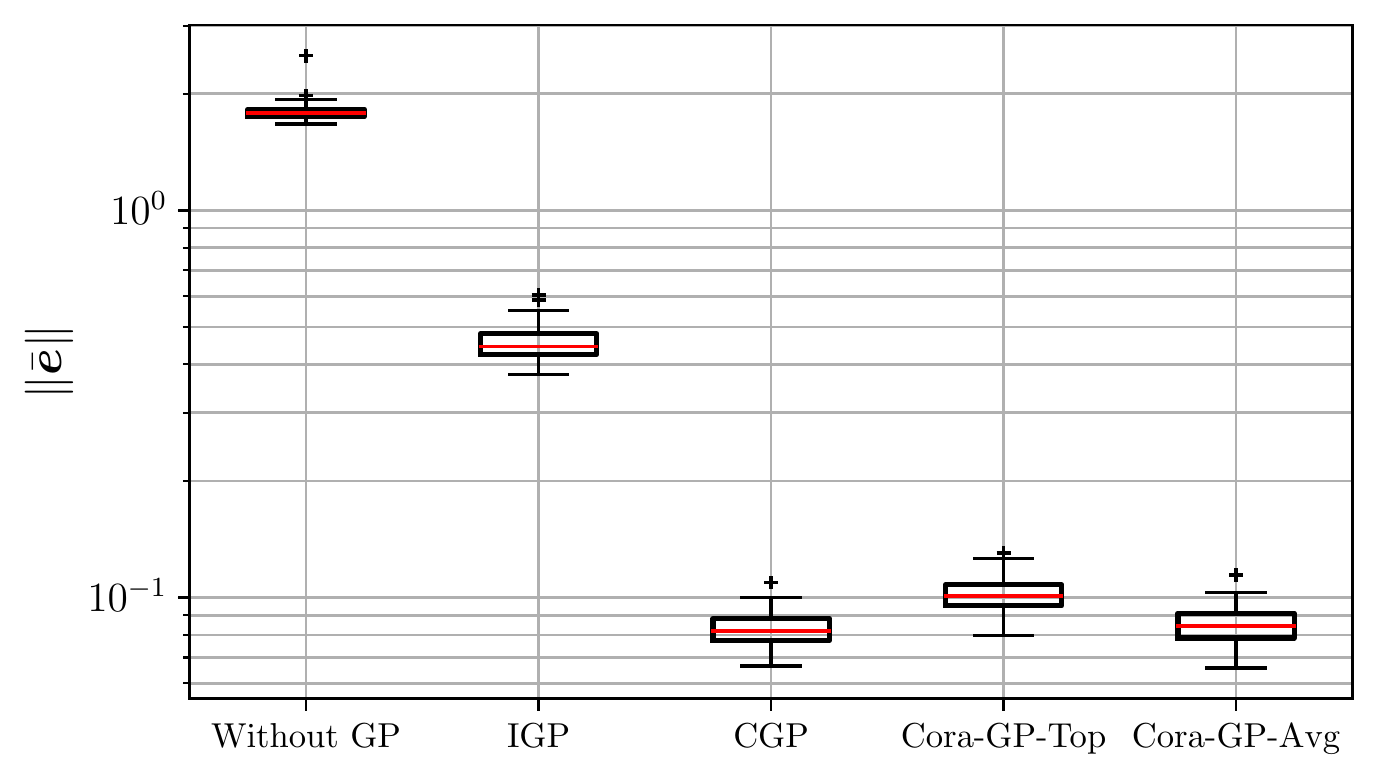}
    \caption{The mean tracking error of different approaches.}
    \label{fig_MCtest}
\end{figure}

\cref{table_time} presents the computational time of obtaining aggregation weights for different approaches over $1000$ computations. The Cora-GP-Top method exhibits an average time reduction of $99.71\%$ compared to CGP, while the Cora-GP-Avg method achieves an even more significant average time reduction of $99.97\%$ compared to CGP. In summary, the proposed Cora-GP framework for cooperative learning in MAS performs similarly against CGP without suffering heavy computation. It enhances the efficiency of obtaining aggregation weights in a crucial aspect when deploying GP-based cooperative learning methodologies. This improvement is of particular significance for applications demanding high prediction rates.
\begin{table}[!t]
\renewcommand{\arraystretch}{1.3}
\caption{Computation time of aggregation weights.}
\centering
\begin{tabular}{l||c||c}
\hline
\bfseries Approach & \bfseries Mean (ms) & \bfseries Median (ms)\\
\hline\hline
Cora-GP-Avg & $\boldsymbol{0.002}$ & $\boldsymbol{0.001}$\\
Cora-GP-Top & $0.019$ & $0.028$\\
CGP & $6.50$ & $6.30$\\
IGP & - & -\\
\hline
\end{tabular}
\label{table_time}
\end{table}

\section{Conclusion}\label{sec_Conclusion}
This paper introduces a distributed consensus tracking control law incorporated with a novel GP-based cooperative learning framework for uncertain ELMASs. The results show a significant stride forward in enhancing the efficiency and efficacy of aggregation weight strategies for cooperative learning. With the proposed learning approaches, the protocol ensures convergence of tracking errors within guaranteed bounds, even when faced with semi-Markov switching communication topologies.

\balance
\bibliography{ref} 

\end{document}